\documentclass[11pt]{article}
\usepackage{geometry}                
\usepackage{graphicx}
\usepackage{amssymb}
\usepackage{url}
\usepackage{latexsym,url}
\usepackage{graphicx}
\usepackage{amssymb}
\usepackage{amsmath}

\usepackage{algorithmicx}
\usepackage{algorithm}
\usepackage{algpseudocode}

\newcommand{\dom}{{\rm dom}}
\newcommand{\N}{\mathbb{N}}
\newcommand{\maps}{\colon}

\newcommand{\ceil}[1]{\left\lceil #1 \right\rceil}
\newcommand{\setto}{\leftarrow}

\def \Q {\mathbb{Q}}

\def \fR{{\mathfrak R}}
\newenvironment{proof}{\trivlist \item[\hskip \labelsep{\bf Proof.}]}
{{\hfill$\Box$}\endtrivlist}

\usepackage{epstopdf}
\newtheorem{thm}{Theorem}
\newtheorem{cor}[thm]{Corollary}

\newtheorem{lem}[thm]{Lemma}

\title{\bf Every Computably Enumerable Random Real Is Provably Computably Enumerable Random}
\author{Cristian S. Calude, Nicholas J. Hay\\
{\small Department of Computer Science}\\
{\small University of Auckland}\\
{\small Private Bag 92019, Auckland, New Zealand}\\
{\small\url{{cristian,nickhay}@cs.auckland.ac.nz}}}

\begin{document}
\maketitle
\begin{abstract} 
We prove that every computably enumerable (c.e.) random real is provable  in Peano Arithmetic (PA) to be c.e.\ random.  A major step in the proof is to show that the  theorem  stating that  ``a  real is c.e.\ and random iff it is the halting probability of a universal prefix-free Turing machine'' can be proven  in PA. Our proof, which is simpler than the standard one, can also be used for the original theorem.

Our positive result  can be contrasted with the case of computable functions, where  not every computable function is provably computable in PA, or even more interestingly,  with the fact that almost all random finite strings are not provably random in PA. 

We also prove two negative results: a) there exists a universal machine whose universality cannot be proved in PA, b) there exists a universal machine $U$ such that, based on $U$, PA cannot prove the randomness of its halting probability.

The paper also includes a sharper form of the Kraft-Chaitin Theorem, as well as a formal proof
of this theorem written with the proof assistant  Isabelle. 
\end{abstract}

\thispagestyle{empty}
\section{Introduction}

A real  in the unit interval  is  {\em computably enumerable} ({\em c.e.})   if it
is the limit of a computable, increasing
sequence of rationals. We identify a real with its infinite binary expansion.  In contrast with
the case of a computable real, whose bits are given by a computable function,
during the process of approximation of a c.e.\ real one may never know how close one is
to the limit.  A real  is (algorithmic)  random if its binary
expansion is an algorithmic random (infinite) sequence \cite{chaitin75,solovaymanu,chaitin87,Ca,DH}. 
 
 A prefix-free machine is a Turing machine, shortly, {\it machine},  from    strings to    strings  whose
domain  is a
prefix-free set.
A  machine  is  universal  if it can simulate every  machine.
Chaitin \cite{chaitin75}  introduced the halting probability
$\Omega_U$ of a   universal   machine $U$, Chaitin's 
Omega number $$\Omega_U = \sum_{U(x) \mbox{  is defined}} \, 2^{-|x|},$$ and proved that {\it $\Omega_U$  is    c.e.\ and random}. As shown by Calude, Hertling, Khoussainov, Wang \cite{CHKW98stacs} and  Ku\v{c}era,
 Slaman \cite{slaman}, (see also \cite{cerand})  there are no other  c.e.\ random reals:
  
 \begin{thm} \label{representation:thm}The set of 
c.e.\ random reals coincides with the set of   halting probabilities of all universal   machines.
\end{thm}

  C.e.\ random reals
have been intensively studied in recent years, with many results summarised in \cite{Ca,DH}.

\begin{thm}[Chaitin \cite{chaitin75}] \label{gregzfc}
 Assume that  {\rm ZFC} (Zermelo-Fraenkel set theory with choice) is arithmetically
sound (that is, any theorem of arithmetic proved by  {\rm ZFC} is true).
Then, for every   universal   machine $U$,  {\rm ZFC} can determine
the value of only finitely many bits of $\Omega_U$, and one can
calculate a bound on the number of bits of
$\Omega_U$ which
 {\rm ZFC} can determine.
 \end{thm}

The real $\Omega_U$ depends on $U$, and so by tuning this choice one gets:
\begin{thm}[Solovay \cite{solovay2k}]
\label{solovay} We can chose a  universal   machine $U$
so that   {\rm ZFC} (if arithmetically sound) cannot determine any  bit  of $\Omega_U$.
\end{thm}

This result was generalised as follows: 
\begin{thm}[Calude \cite{Incompl}] \label{crisincompl} Assume that  {\rm ZFC} is arithmetically sound. Let $i \ge 1$ and consider the c.e.~random real 
$ \alpha = 0. 1^{i-1}0
\alpha_{i+1}\cdots $
Then, we can effectively construct  a universal   machine $U$ (depending upon
 {\rm ZFC} and $\alpha$)  such that  {\rm PA} (Peano Arithmetic) proves the universality of $U$,
  {\rm ZFC} can  determine at most $i-1$ initial bits of $\Omega_U$ and $\alpha = \Omega_U$.
  \end{thm}

The proof of Theorem~\ref{crisincompl}  in \cite{Incompl} starts by fixing a universal  
machine $V$ such that the universality of $V$ is provable in  PA 
and $\Omega_V = \alpha$.  Solovay \cite{solovayemail} observed that  ``it is by no means 
evident that there is a universal prefix-free machine whose universality is 
provable in  PA  and whose halting probability is $\alpha$''.   \begin{quote} Let $\alpha \in  (0,1)$ be c.e.\ and random. Is there any {\em representation}
of $\alpha$ for which  PA  can \emph{prove} 
that $ \alpha$ c.e.\ and random?
\end{quote}
 We  give an affirmative answer to this question.  
A major step in the proof is to show that Theorem~\ref{representation:thm}  can be proved   in PA. Our proof, which is simpler than the standard one, can be  used also for the original theorem.

The paper also includes a sharper form of the Kraft-Chaitin Theorem, as well as a formal  proof
of this theorem written with the proof assistant  Isabelle. 

 In what follows proofs will be written in Solovay's style \cite{solovay2k}. All necessary steps are presented in sufficient detail to leave the remaining formalisation routine. The
 formalisation of the  Kraft-Chaitin Theorem is  presented with full details, and then a sketch of the formal proof in Isabelle is  discussed.

 The paper is organised as follows. Sections 2 and 3 present all  facts on formal provability and Algorithmic Information Theory needed for this paper.  The   Kraft-Chaitin Theorem is presented in Section~4. Section~5 presents three ways to prove randomness, using Martin-L\"of tests,  prefix-free complexity, and Solovay representation formula.  In Section~6  we revisit Chaitin's Theorem on the randomness of the halting probability of a  universal  machine.  In Section~7   we prove that 
a real $\alpha \in (0,1)$ is provably Chaitin-random iff it is provable that $\alpha = \Omega_U$ for some provably universal  machine $U$ (see Theorem~\ref{thm:representation}). 
  In Section~8 we prove our main theorem: 
 every c.e.\ random real is provably random (Theorem~\ref{thm: cerandprov}).
 In Section~9 we construct a universal machine $U$ based on which PA cannot prove the randomness of its halting probability.
 Section~10 presents a formal proof of the Kraft-Chaitin Theorem written with  Isabelle. The  final Section~11  includes a few general remarks.

\section{Provability}

By $\mathcal{L}_{A}$ we denote the first-order language   of arithmetic whose non-logical  symbols consist of  the constant symbols 0 and 1, the binary relation symbol $<$ and  two binary function symbols $+$ (addition) and $\cdot$ (multiplication). Peano Arithmetic (see  \cite{RK}, shortly, PA) is  the first-order theory given by a set of 15 axioms defining discretely ordered rings,  together with induction axioms for each formula $\varphi (x,y_{1},\ldots ,y_{n})$ in $\mathcal{L}_{A}$:
\[\forall \overline{y}(\varphi(0,\overline{y}) \wedge \forall x(\varphi(x,\overline{y}) \rightarrow \varphi(x+1,\overline{y})) \rightarrow \forall x (\varphi(x,\overline{y})).\]
\if01
This schema retains as much as possible from  the second-order induction axiom:
\[\forall X (0\in X \wedge \forall x(x\in X\rightarrow x+1 \in X) \rightarrow \forall y (y\in X)),\]
where $X$ ranges over all subsets of the domain, and $x,y$ range over elements of this domain, but
avoids quantification over sets of natural numbers, which is impossible in first-order logic. 

In PA it is not possible  to say that any set of natural numbers containing 0 and closed under successor is the entire set of natural numbers, but  that any definable set of natural numbers has this property. The induction schema includes one instance of the induction axiom for every definition of a subset of the naturals.
\fi
The structure {\bf N} whose domain is  the set of naturals $\N=\{0,1,2,\ldots\}$,  where the symbols in   $\mathcal{L}_{A}$ have the obvious interpretation, satisfies the axioms of PA; this is the standard model for PA. There are non-standard models  of PA that are not isomorphic to
{\bf N}. If $M$ is a structure for $\mathcal{L}_{A}$ and $\varphi(\overline{x})$ is an $\mathcal{L}_{A}$-formula with free-variables  $\overline{x} = (x_{1}, \ldots , x_{n})$
and $\overline{a} = (a_{1}, \ldots ,  a_{n}) \in M$, then we write $M \vDash \varphi(\overline{a})$ to mean that ``$\varphi$ is true in $M$ when each variable $x_{i}$ is interpreted by $a_{i}$''. We blur the distinction between $n$ and  the closed term of $\mathcal{L}_{A}$, $(\cdots (((1+1)+1)+1)+ \cdots 1)$, ($n$ times).

A formula $\theta (\overline{x})$ of $\mathcal{L}_{A}$ is $\Delta_{0}$ if all its quantifiers are bounded. A formula $\psi (\overline{x})$ of $\mathcal{L}_{A}$ is $\Sigma_{1}$ if it is of the form $\psi (\overline{x}) = \exists y \theta (\overline{x}, y)$ with $\theta (\overline{x}, y) \in \Delta_{0}$; $\psi (\overline{x})$ of $\mathcal{L}_{A}$ is $\Pi_{1}$ if it is of the form $\psi (\overline{x}) = \forall y \theta (\overline{x}, y)$ with $\theta (\overline{x}, y) \in \Delta_{0}$. 

By PA $\vdash \theta$ we mean ``there is a proof in PA for $\theta$''.  It is useful to know  that PA proves the {\it least number principle}:
${\rm PA} \vdash  \forall y (\exists x \varphi (x, y) \rightarrow \exists z (\varphi (z,  y) \wedge \forall w < z \, \neg \varphi (w,  y))),$
for each formula $\varphi (x,  y) $ of $\mathcal{L}_{A}$.

An important link between computability  and provability is given by the following results.

\begin{thm}  A partial function from $\N$ to $\N$  is partial computable iff its graph is equivalent to a  $\Sigma_{1}$ $\mathcal{L}_{A}$-formula.
\end{thm}

\begin{cor} \label{cesets}A set $A \subset \N^{k}$ is computably enumerable (c.e.) if there is a  $\Sigma_{1}$ $\mathcal{L}_{A}$-formula $\varphi(\overline{x})$ such that for all $\overline{x}\in \N^{k}$, $\overline{x}\in A $ iff
{\bf N}  $  \vDash  \varphi(\overline{x})$.
\end{cor}

 A total function $f: \N^{k} \rightarrow \N$ is represented in PA if there is an $\mathcal{L}_{A}$-formula $\theta(\overline{x})$ such that for all $\overline{n} \in \N^{k}$:\\[-4ex]
\begin{enumerate}
\item PA $\vdash \exists ! y \theta (\overline{n}, y)$, and
\item if $k=f(\overline{n})$ then PA $\vdash \theta (\overline{n}, k)$.
\end{enumerate}
\mbox{}\\[-4ex]
(Here  $\exists ! $ means ``there exists a unique''.)
One can show that every total computable function is represented by a $\Sigma_{1}$-formula of PA  \cite{RK}.

A function $f: \N \rightarrow \N$ is {\it provably computable} \cite{PF,RK} if there exists a $\Sigma_{1}$-formula of PA $\varphi(x,y)$ such that:\\[-4ex]
\begin{enumerate}
\item $\{(n,m) \mid m,n, \in \N, f(n)=m\} = \{(n,m) \mid {\bf N} \vDash  \varphi (n,m)\}$,
\item PA $\vdash \forall x \exists ! y \varphi (x, y)$.
\end{enumerate}
\mbox{}\\[-4ex]
In view of  Corollary~\ref{cesets}, any provably computable function has a c.e.\ graph, so it is total and computable. These functions can be viewed as computable functions whose totality is proved by PA. 

\begin{thm}[\cite{RK}] Every primitive recursive function is provably computable, but there exist computable functions which are not provably computable in {\rm PA}.
\end{thm}

If $f$ is computable but not provably computable in PA, then the statement ``$f$ is total'' is true but unprovable in PA. 
In contrast with the case of computable functions, {\it c.e.\ sets are provably enumerable} \cite{PF} (because every non-empty c.e.\ set  can be enumerated by a primitive recursive function, \cite{BL}, p. 138).

In what follows all computations will be implemented   by primitive recursive functions. Hence, we will work with a special type of $\Sigma_{1}$ formulae.
By abuse of language we  say that a formula of PA is $\Sigma_{1}^{0}$ if it has the form
$\exists x P(x)$, for some primitive recursive predicate $P(x)$. A formula of PA is $\Pi_{2}^{0}$ 
if it has the form
$\forall x \exists y P(x,y)$, for some primitive recursive predicate $P(x,y)$.

Our metatheory is  {\rm ZFC}. We fix a (relative) interpretation of PA in  {\rm ZFC}. Each formula of $\mathcal{L}_{A}$ has a translation into a formula of  {\rm ZFC} determined by the interpretation of PA in  {\rm ZFC}. By abuse of language we shall use the phrase ``sentence of arithmetic''  to mean a formula with no free variables of  {\rm ZFC}  that is the translation of some formula of PA.
We assume that  {\rm ZFC} is 1--consistent, that is, if it proves a $\Sigma_{1}^{0}$ sentence then that sentence is true (in the standard model of PA). 
\begin{thm} [Solovay \cite{solovay2k}] Every $\Pi_{2}^{0}$  sentence proved by  {\rm ZFC} is true.
\end{thm} 
As a consequence, it follows that {\it if $U$ is a   machine which} {\rm PA} {\it can prove universal and  {\rm ZFC} can prove  the sentence ``the i-th digit of $\Omega_{U}$ is $k$'', then
the sentence is true}. Whenever we  talk about the provability of a sentence of arithmetic
we mean that PA proves its corresponding translation  formula.

If there is a proof  in PA for statement $A$ we  say that $A$  provable in PA.  
We say that $A$ is provably $P$ (where $P$ is a property) if  the statement  ``$A$ has $P$''
is provable in PA.



 
\section{Algorithmic Information Theory: Some Definitions and Results}
 \label{ait}
{\em All reals are in the unit interval.}  A c.e.\ real $\alpha$ is represented by an increasing computable sequence of rationals converging to $\alpha$. We  blur the distinction between the real $\alpha$ and the infinite base-two expansion of
$\alpha$, i.e.\  the infinite    sequence $\alpha_{1} \alpha_{2}\cdots \alpha_{n}\cdots $ ($\alpha_{n}\in\{0,1\})$ such that $ \alpha = 0.\alpha_{1} \alpha_{2}\cdots \alpha_{n}\cdots $ By
$\alpha(n)$ we denote the  string of length $n$, $\alpha_{1} \alpha_{2}\cdots \alpha_{n}$.
 
 The set of (bit) strings is denoted by $\Sigma^{*}$; $\varepsilon$ denotes the empty string. If $s$ is a    string then $|s|$ denotes the length of $s$. We import the theory of computability from natural numbers to    strings by fixing the  canonical bijection between $\Sigma^{*}$ and $\N$ induced by the
linear order $s<t$ if $|s| < |t|$ or   $|s| = |t|$ and $s$ lexicographically precedes $t$.

 A   machine $U$ is {\em universal}  if for every   machine $V$
there is a constant $c$ (depending upon $U$ and $V$) such that for all
   strings $s,t$, if $V(s)=t$, then $U(s')=t$ for some    string $s'$
of length $|s'| \leq |s| + c$.  The domain of $U$ is the set $\{x \in \Sigma^{*} \mid U(x) \mbox{  is defined}\}$. The Omega number $\Omega_{U} =  \sum_{x \in \dom{U}} 2^{-|x|}$ is halting probability of $U$. The {\em prefix-free complexity} of the string
$x\in\Sigma^*$ (relatively to the   machine $C$)
is  $H_C(x)=\min \{|y| \mid  y \in \Sigma^*, \ C(y)=x\}$
($\min \emptyset = \infty$). If $U$ is  a universal   machine,  then for every    machine we can effectively construct a constant $c$ (depending on $U$ and $C$) such that
 $H_U (x) \leq H_C (x) + c$, for all $x$.
 
 A real $\alpha$ is {\em Chaitin-random} if there exists a universal   machine $U$ and constant $c$ such that for all $n\ge 1$, $H_{U}(\alpha (n))\ge n-c$.

A c.e.\ open set is a c.e.\ union  of intervals with rationals endpoints $[a,b)$ and $\mu$ is Lebesgue measure. If $S$ is a prefix-free set, then $\mu (S)$ denotes the Lebesgue measure of the cylinder denoted by $S$, i.e.\ all reals whose infinite binary expansions have a prefix in $S$.
To the    string $x$ we associate the interval $[0.x, 0.x + 2^{-|x|})$ of measure $2^{-|x|}$.   A Martin-L\"{o}f 
test (shortly, ML test)  $A$ is a uniformly c.e.\ sequence of c.e.\ open sets $A = (A_n)$  such that for all $n\ge1$, $\mu(A_n) \leq 2^{-n}$.  A real $\alpha$ is {\em Martin-L\"of-random} (shortly, {\em ML-random})  if for every ML test $A$ there exists an $i$ such that $\alpha \not\in A_{i}$. A classical theorem states that {\em a real is Chaitin-random iff
  it is  ML-random} \cite{chaitin87,Ca}.

Note that Chaitin and Martin-L\"of definitions apply to any real. In the special case of c.e.\ reals the following Solovay representation formula stated in \cite{solovayemail} is used: A real $\alpha$ is {\em c.e.\ and random} if there exists a universal  machine $U$, an integer $c>o$ and a c.e.\ real $\gamma>0$ such that $\alpha = 2^{-c} \cdot \Omega_{U} + \gamma$ (see Lemma~\ref{lem:solovayrepresentation}).

\section{Kraft-Chaitin Theorem Revisited}

 We start by showing that PA can prove the Kraft-Chaitin Theorem \cite{chaitin87,Ca}.
 
\begin{thm}
\label{thm:kc}
Suppose $(n_i, y_i)_i \in \N \times \Sigma^*$ is a primitive recursive enumeration of ``requests'' which provably satisfies
$
	\sum_{i} 2^{-n_i} \le 1.
$
Then there exists a provably prefix-free machine $M$ and a primitive recursive enumeration $(x_i)_i$ of $\dom(M)$ such that the following is provable in {\rm PA}:\\[-4ex]
\begin{enumerate}
\item	$\mu( \dom(M) ) = \sum_{i} 2^{-n_i}$,
\item	$|x_i| = n_i$ for all $i\in\N$,
\item	$M(x_i) = y_i$ for all $i\in\N$.
\end{enumerate}
\end{thm}
\begin{proof}
Algorithm \ref{alg:kc} below enumerates the graph of $M$.  Intuitively, $S_i$ keeps track of the tree of prefixes we haven't allocated yet.  To start with we have allocated nothing, so $S_0 = \{\epsilon\}$.  At each step we want a string (node) of a given length (depth) $n_i$.  The program selects the deepest leaf it can, then creates the smallest number of new leaves to create the node we need.\\[-1.7ex]

\begin{algorithm}[htb]
\caption{}\label{alg:kc}
\begin{algorithmic}[1]

\State $S_0 = \{\epsilon\}$, $T_0 = \emptyset$, $r_0 = 0$, $i \setto 0$.
\Loop
	\State	Let $s_i$ be the longest element of $S$ of length at most $n_i$.  If no such string exists, terminate.
	\If{$|s_i|=n_i$}
		\State	$S_{i+1} = (S_i \setminus \{s_i\})$.
	\Else
		\State	$S_{i+1} = (S_i \setminus \{s_i\}) \cup \{ s_i1, s_i01, s_i0^21, \ldots, s_i0^{n_i-|s_i|-1}1 \}$.
	\EndIf
	\State	Define $M(s_i0^{n_i-|s_i|}) = y_i$.
	\State	$T_{i+1} = T_i \cup \{s_i0^{n_i-|s_i|}\}$.
	\State	$r_{i+1} = r_i + 2^{-n_i}$.
	\State	$i\setto i+1$.
\EndLoop

\end{algorithmic}
\end{algorithm}

Examining Algorithm~\ref{alg:kc}, it is clear that the sequence $x_i = s_i0^{n_i-|s_i|}$ is a primitive recursive enumeration of $\dom(M)$, and  whenever $x_i$ is defined we have $M(x_i) = y_i$ and $|x_i| = n_i$.  It remains to show that $x_i$ is defined for all $i\in\N$ (i.e. the program never terminates), that $\dom(M)$ is prefix-free, and $\mu(\dom(M)) = \sum_i 2^{-n_i}$.

It suffices to establish, for all $i$, the following invariants: \\[-4ex]
\begin{enumerate}
\item	$S_i \cup T_i$ is prefix-free (which implies that $S_i$ and $T_i$ individually are prefix-free),
\item	$\mu(S_i \cup T_i) = 1$,
\item	$\mu(T_i) = r_i$,
\item	$\sum_{j\ge i} 2^{-n_j} \le \mu(S_i)$,
\item	If $2^{-n} \le \mu(S_i)$, then $S_i$ contains a string of length at most $n$ (equivalently,  $S_i$ contains strings of distinct length).
\end{enumerate}

The base case is trivial.  For the inductive step, first observe that line 3 of Algorithm~\ref{alg:kc} doesn't terminate since $2^{-n_i} \le \mu(S_i)$ by invariant 5.  We see that
\[
	S_{i+1} \cup T_{i+1} = \left( (S_i \cup T_i) \setminus \{s_i\} \right) \cup \{ s_i1, s_i01, s_i0^21, \ldots, s_i0^{n_i-|s_i|-1}1, s_i0^{n_i-|s_i|} \}
\]
which is prefix-free establishing invariant 1.  From this we can see invariant 2 holds:
$
	\mu(S_{i+1} \cup T_{i+1}) = \mu(S_{i} \cup T_{i}) = 1.
$
Next observe invariant 3 holds too:
$
	\mu(T_{i+1}) = \mu(T_i) + 2^{-n_i} = r_{i+1},
$
which implies that
$
	\mu(S_{i+1}) = \mu(S_i) - 2^{-n_i}.
$
From this follows invariant 4:
$
	\sum_{j\ge i+1} 2^{-n_j} \le \mu(S_{i+1}).
$
Finally, since $s_i$ is the longest string of length at most $n_i$ in $S_{i}$, and we add strings of distinct length between $s_i+1$ and $n_i$ to $S_i$ to form $S_{i+1}$, we see that $S_{i+1}$ consists of strings of distinct lengths.  This establishes invariant 5.
\end{proof}

\section{Randomness and Provability}

 In this section we discuss three forms of provability for randomness.

  There are two ways to represent a c.e.\ real number $\alpha \in (0,1)$ in PA: 1) by giving an increasing 1-1 primitive recursive function that enumerates a  c.e.\ prefix-free set of strings $\{s_{i}\}$ such that
 $\alpha = \sum_{i} 2^{-|s_{i}|}$, 2) by giving an increasing primitive recursive sequence $(a_i)_i $ of rationals in the unit interval whose limit is $\alpha$.
 It is clear that given the representation 1) one can effectively get
 the representation 2). The converse is also true.

 \begin{lem}
\label{lem:repce}
Let $\alpha$ be a c.e.\  real defined by the increasing primitive recursive sequence $(a_i)_i $ of rationals.
Then there is a primitive recursive sequence $(n_i)_{i}$ of natural numbers such that {\rm PA} proves
$
	\sum_{i} 2^{-n_i} = \alpha.
$
\end{lem}

\begin{proof}

Without loss of generality assume $a_1>0$.  Define the primitive recursive sequences $(r_i)_i \in \Q$ and $(n_i)_i \in \N$ by $r_0=0$ and for $i\ge 1$ by 
\[	n_i = \ceil{-\log_2 (a_i-r_{i-1})}, 
	r_i = r_{i-1} + 2^{-n_i}.\]
Since $(r_i)_i$ is strictly increasing we can establish by induction the inequality $r_{i-1} < a_i$ for all $i$, making the logarithm well-defined.  By construction we have
$
	\sum_{i} 2^{-n_i} = \lim_{i\to\infty} r_i.
$
Define $\beta = \lim_{i\to\infty} r_i$.
Since 
$
	-\log_2 (a_i-r_{i-1}) \le n_i \le -\log_2 (a_i-r_{i-1})+1
$
we have
$
	 (a_i+r_{i-1})/2  \le r_i  \le  a_i.
$
Taking the limit we see that $(\alpha+\beta)/2 \le \beta \le \alpha$ establishing our result.
%
\end{proof}

\begin{cor}
Let $\alpha$ be a c.e.\  real defined by the increasing primitive recursive sequence $(a_i)_i $ of rationals.
Then there is a  machine $M$ such that {\rm PA} proves that $\alpha =\mu(\dom(M))$.
\end{cor}

\begin{proof} Use Lemma~\ref{lem:repce} and Theorem~\ref{thm:kc}.
\end{proof}

In what follows a c.e.\ real is given by one of the above representations.

  A c.e.\ real $\alpha$ is {\em provably Chaitin-random} if there exists a provably universal   machine $U$ and constant $c$ such that {\rm PA}  proves that for all $n\ge 1$, $H_{U}(\alpha (n))\ge n-c$.
  A c.e.\ real $\alpha$ is {\em provably  ML-random}  if for every set $A$ which {\rm PA}  proves to be a ML test and {\rm PA}  proves that   there exists an $i$ such that  $\alpha \not\in A_{i}$.

  The classical theorem that states that {\em a real is Chaitin-random iff
  it is  ML-random} is provable in PA. 
 However, for the goal of this paper only one implication is needed:

 \begin{thm}\label{thm:chaitintoML}
 Every c.e.\ provably Chaitin-random real is provably  ML-random.
 \end{thm}
\begin{proof}
 Take a c.e. real $\alpha$,  a  machine $U$ which is provably universal  and a natural $c>0$ such that  {\rm PA} proves that for all $n\ge 1$,
 $H_{U}(\alpha (n)) \ge n-c$.

We wish to prove that for every   $A= (A_{n})$   which {\rm PA} proves to be a ML test there exists an $i$ such that {\rm PA} proves that $\alpha \not\in A_{i}$. 
Following the proof of  Proposition~6.3.4 in \cite{Ca} it follows that PA proves the existence of     a c.e.\ set $S\subset \Sigma^{*}\times \N$
such that each $S_{i} = \{x \in \Sigma^{*} \mid (x,i)\in S\}$ is prefix-free, and by taking $A_{i} = \{\beta \mid \beta(m)\in S_{i}, \mbox{  for some } m\ge 1\}$ we get $\mu(A_{i}) =
\sum_{s\in S_{i}} 2^{-|s|}$.
 
Let $g \maps \N \to \Sigma^* \times \Sigma^*$ be a 1-1 primitive recursive enumeration of the graph of $U$.  Denote by $\pi_i \maps \Sigma^* \times \Sigma^* \to \Sigma^{*}$ for $i=1,2$ the projection functions and  $f(i) = \pi_1(g(i))$ is a 1-1 primitive recursive enumeration of $\dom(U)$.
Note that $H_U(x)$ can be expressed in {\rm PA}:
$H_U(x) = \min_i \{ |v|: U(v) = x \} 
		= \min_i \{ |\pi_1(g(i))| \mid  \pi_2(g(i)) = x \}$,

We have:
 \[\sum_{n\ge 2} \sum_{s \in S_{n^{2}} }  2^{-\lfloor |s| -n\rfloor} = \sum_{n\ge 2} 2^{n}\mu(S_{n^{2}}) \le \sum_{n\ge 2} 2^{n}2^{-n^{2}}\le 1.\]
 
 We can now use Theorem~\ref{thm:kc}:  There exists a provably prefix-free  machine $M$
 such that: $ \mu(\dom(M)) = \sum_{n\ge 2} \sum_{s \in S_{n^{2}} }  2^{-\lfloor |s| -n\rfloor}$, $\dom(M) = \{r_{n,s} \mid n\ge 2, s\in S_{n^{2}}, |r_{n,s} |= |s|-n\}, M(r_{n,s} )=s.$ Since $U$ is provably universal there is a constant $d\ge 1$ such that for all strings $x, H_{U}(x) \le H_{M}(x)+d$, so in particular, {\rm PA} proves that for all $n\ge 2$, if $s\in S_{n^{2}}$, then
 $H_{U}(s) \le H_{M}(s)+d \le |s|-n+d < |s|-n+d+1.$
 
 We are now in a position to find a natural $n\ge 2$ such that {\rm PA} proves that $\alpha \notin A_{n^{2}}$ showing that $\alpha$ is provably ML-random.
   Note  that for $n\ge 2$, {\rm PA} proves  $\alpha \notin A_{n^{2}}$ iff for all $m\ge 1$,
  {\rm PA} proves that $\alpha(m) \notin S_{n^{2}}$.   For all $m\ge 1$,  {\rm PA} proves that  $\alpha(m) \in S_{n^{2}}$ implies $H_{U}(\alpha(m))  < m-n+d+1.$ Hence, for $2 \le n < c+d+1$, {\rm PA} proves that $H_{U}(\alpha(m)) \ge m-c$ implies $\alpha(m) \notin S_{n^{2}}$, so because $\alpha$ is provably Chaitin-random {\rm PA} proves that $\alpha \notin A_{n^{2}}$.
 \end{proof}

{\bf Comment } The above proof shows that that
$T^{U}_{m} = \{\beta \mid H_{U}(\beta (n)) < n-m, \mbox{  for some  } n\ge 1\}$, where $U$ is a provably universal  machine, is a
provably ML test such that for all $n\ge 2$ and provably  ML test $A$ there exists  $d>0$ such that {\rm PA} proves  the inclusion $A_{n^{2}} \subset T^{U}_{n-d-1}$, i.e. $(T^{U}_{m})^{}_{m}$ is a provably universal ML test.

\medskip

To be able to complete our program we need to choose a specific representation for a c.e.\ and random real which can be ``understood'' by {\rm PA} and, even more importantly, {\rm PA} can extract from it a proof of the
randomness of the real (c.e.\ is obvious). First we  work with Solovay representation formula discussed at the end of Section~\ref{ait}.

   A real $\alpha$ is {\em c.e.\ and provably random} if there exists a representation of $\alpha$ in the form 
\begin{equation}
\label{solovayformula}
\alpha = 2^{-c}\cdot\Omega_V + \gamma,
\end{equation}
  where $V$ is a provably universal  machine, $c\ge 0$ is an integer and $\gamma>0$ is a provably c.e.\  real.  Theorem~\ref{lem:solovayrepresentation}  shows that all c.e.\  random reals have a representation of this form.
In detail, {\rm PA}  receives an algorithm for a machine  $V$, a proof that $V$ is prefix-free and universal, an
 integer $c\ge 0$ and a computable increasing sequence of rational converging to a real $\gamma >0$. The goal is to prove that {\rm PA} can use this information to prove that $\alpha = 2^{-c}\cdot\Omega_V + \gamma$ is c.e.\ and random.

\section{Chaitin's Theorem Revisited}
  Chaitin \cite{chaitin75} proved that the halting probability of a universal  machine is Chaitin-random.
  This theorem is provable in {\rm PA}:

\begin{thm}
\label{thm:chaitin}
Suppose $U$ is provably universal.  Then 
$
	\Omega_U = \sum_{p \in \dom(U)} 2^{-|p|}
$
is provably Chaitin-random.
\end{thm}
\begin{proof}
Let $g \maps \N \to \Sigma^* \times \Sigma^*$ be a 1-1 primitive recursive enumeration of the graph of $U$.  Denote by $\pi_i \maps \Sigma^* \times \Sigma^* \to \Sigma$ for $i=1,2$ the projection functions and  $f(i) = \pi_1(g(i))$ is a 1-1 primitive recursive enumeration of $\dom(U)$.
Recall that $H_U(x)$ can be expressed in {\rm PA}.
\if01
\begin{align*}
	H_U(x) &= \min_i \{ |v|: U(v) = x \} \\
		&= \min_i \{ |\pi_1(g(i))| : \pi_2(g(i)) = x \}.
\end{align*}

\[H_U(x) = \min_i \{ |v| \mid U(v) = x \} \\
		= \min_i \{ |\pi_1(g(i))| \mid \pi_2(g(i)) = x \}.\]
		\fi
Define the  machine $M$ by $M(0^{|x|}1x) = x$.  Since $U$ is provably universal, there is a $c$ such that for all $x$, 
$
	H_U(x) \le H_M(x) + c = 2|x| + c+1.
$
This shows that $H_U(x)$ is provably total and  $U$ is provably onto.

Define the primitive recursive sequence of rationals
$
	\omega_k = \sum_{i}^{k} 2^{-|\pi_1(g(i))|}
$
and notice that $(\omega_k)_{k}$  is provably strictly increasing;  $\Omega_U$ is, by definition, the limit of this sequence.

Define $C(x) = v$  if there exist $t$,$j$ such that\\[-4ex]
\begin{enumerate}
\item	$\pi_1(g(j)) = x$ (i.e. $x\in\dom(U)$),
\item	$t$ is the least such that $0.\pi_2(g(j)) \le \omega_t$ (i.e. $0.U(x)\le\omega_t$),
\item	$v$ is the lexicographically least string such that $v \neq \pi_2(g(s))$ for all $1\le s\le t$.
\end{enumerate}

This defines a provably prefix-free machine.  Observe that if $C(x)$ is defined and $U(x)=U(x')$ then $C(x)=C(x')$.  From this we can establish that whenever $C(x)$ is defined we have
$
	H_C(C(x)) \le H_U(U(x)).
$
As $U$ is provably universal,  there exists an $a$ such that for all $y$,
$
	H_U(y) \le H_C(y) + a
$
is provable in {\rm PA}.

Denote by $\Omega_i$ the $i$th  digit  of $\Omega_U$.  Since $U$ is provable onto, for each $n$ there exists a string $x_n$ such that
$
	U(x_n) = \Omega_1 \cdots \Omega_n.
$
Since $0.\Omega_{1}\cdots\Omega_n < \Omega_U$ we know that $C(x_n)$ is defined.  Let $t$ be  the least natural (found when evaluating $C(x_n)$) such that
$
	0.U(x_n) \le \omega_t < \Omega_U < 0.U(x_n) + 2^{-n}.
$
 The inequality
$
	\sum_{i\ge t+1} 2^{-\pi_1(g(i))} \le 2^{-n}
$ is easy consequence,
 so for all $i\ge t+1$ we have $|\pi_2(g(i))| \ge n$.  Since $C(x_n)$ equals $g(i)$ for some $i\ge t+1$ by construction, we have that for all $n$
\if01
\begin{align*}
	n 	&\le H_U(C(x_n)) \\
		&\le H_U(U(x_n)) + a \\
		&= H_U(\Omega_1 \cdots \Omega_n) +a
\end{align*}
\fi
\[n 	\le H_U(C(x_n)) \\
		\le H_U(U(x_n)) + a \\
		= H_U(\Omega_1 \cdots \Omega_n) +a\]
is provable in {\rm PA}.  That is, $\Omega$ is provably Chaitin-random.
\end{proof}

 From Theorem~\ref{thm:chaitin} we deduce that {\rm PA} can prove the implication:  ``if $U$ is a provably  universal   machine, then $\Omega_{U}$ is Chaitin-random.''  We know that every
  c.e.\ and random real is the halting probability of a universal  machine, but we need more: {\it Can any c.e.\ and random real be represented
as the halting probability of a {\rm provably}  universal   machine?}  First we have to check   whether every  universal   machine
is provably  universal.

\begin{thm}
\label{thm:nonprovunivmachine} There exist a provably universal  machine and  a  universal   machine that is not provably universal.
\end{thm}
\begin{proof} The set of all  provably prefix-free machines is c.e., so if 
$(M_i)_{i}$ is a computably enumeration of provably prefix-free machines, then the machine $U$ defined by $U(0^i1x) = M_i(x)$
is a provably universal   machine.

Let $(f_{i})_{i}$ be a c.e.\ enumeration of all primitive recursive functions $f_{i}: \N \rightarrow
\Sigma^{*}$ and $(T_{i})$ a c.e.\ enumeration of all  machines. Fix a universal   machine $U$ and consider the computable function $g: \N \rightarrow \N$ such that:

\[ T_{g(i)}(x) = \left\{ \begin{array}{ll}
 U(x), & \mbox{\rm if for some $j>0,  \#\{f_{i}(1), f_{i}(2), \ldots ,f_{i}(j)\}> |x|$}, \\
\infty, & \mbox{\rm otherwise} \,.
  \end{array} \right.\]

For every $i$, $T_{g(i)}$ is a   universal machine iff $f_{i}(\N)$ is infinite (if $f_{i}(\N)$ is finite, then so is  $T_{g(i)}$). Since the set
of all indices of primitive recursive functions with infinite range is not c.e.\
it follows that there is an $i$ such that  {\rm PA} cannot prove that  $ T_{g(i)}$ is  universal.
\end{proof}

Theorem~\ref{thm:nonprovunivmachine} does not imply a negative answer for the previous question; in
fact, Corollary~\ref{cor:secondrep}  shows that the answer is affirmative.  Theorem~\ref{thm:nonprovunivmachine} produces examples of true and unprovable (in PA) statements of the form ``$V$ is universal''.

\section{Provably C.E.\ Random   Reals}

In this section we sharpen Theorem~\ref{representation:thm}  by proving
 that  a  real is provably c.e.\ and Chaitin-random  iff it is provable that the real is the halting probability of a provably universal  machine.

\if01
\begin{lem}
\label{lem:repce}
Let $\alpha$ be a c.e.\  real defined by the increasing primitive recursive sequence $(a_i)_i $ of rationals.
Then there is a primitive recursive sequence $(n_i)_{i}$ of natural numbers such that {\rm PA} proves
\[
	\sum_{i} 2^{-n_i} = \alpha.
\]
\end{lem}

\begin{proof}

Without loss of generality assume $a_1>0$.  Define the primitive recursive sequences $(r_i)_i \in \Q$ and $(n_i)_i \in \N$ by $r_0=0$ and for $i\ge 1$ by 
\begin{align*}
	n_i &= \ceil{-\log_2 (a_i-r_{i-1})}, \\
	r_i &= r_{i-1} + 2^{-n_i}.
\end{align*}
Since $(r_i)_i$ is strictly increasing we can establish by induction $r_{i-1} < a_i$ for all $i$, making the logarithm well-defined.  By construction we have
\[
	\sum_{i} 2^{-n_i} = \lim_{i\to\infty} r_i.
\]
Define $\beta = \lim_{i\to\infty} r_i$.
Since 
\[
	-\log_2 (a_i-r_{i-1}) \le n_i \le -\log_2 (a_i-r_{i-1})+1
\]
we have
\[
	 (a_i+r_{i-1})/2  \le r_i  \le  a_i.
\]
Taking the limit we see that $(\alpha+\beta)/2 \le \beta \le \alpha$ establishing our result.
%
\end{proof}

\begin{cor}
Let $\alpha$ be a c.e.\  real defined by the increasing primitive recursive sequence $(a_i)_i $ of rationals.
Then there is a  machine $M$ such that {\rm PA} proves that $\alpha =\mu(\dom(M))$.
\end{cor}
 \fi

According to 
Solovay \cite{solovaymanu}
a c.e.\ real $\alpha$ {\em Solovay dominates } a c.e.\ real $\beta$  (we write
$\beta \leq_S \alpha$)  if 
there are two computable, increasing  sequences $(a_i)_i$
and $(b_i)_i$ of rationals and a constant $c$ with
$\lim_{n\rightarrow \infty} a_n = \alpha$,
$\lim_{n\rightarrow \infty} b_n =\beta$, and
$c (\alpha - a_n) \geq \beta - b_n$, for all $n$.

For  c.e.\ reals $\alpha, \beta$, {\rm PA} proves
$\beta \leq_S \alpha$  if 
there are two primitive recursive, increasing  sequences $(a_i)_i$
and $(b_i)_i$ of rationals and a constant $c$  such that {\rm PA} proves
$\lim_{n\rightarrow \infty} a_n = \alpha$,
$\lim_{n\rightarrow \infty} b_n =\beta$, and
$c (\alpha - a_n) \geq \beta - b_n$, for all $n$.



\begin{thm}\label{thm:dominates}
If $\alpha$ is c.e.\ and provably 
ML-random, and $\beta$ is c.e., then $\beta \le_S \alpha$ is provable in {\rm PA}.
\end{thm}

\begin{proof}




Let $(a_i)_i$ and $(b_i)_i$ be primitive recursive sequences of rationals with limits $\alpha$ and $\beta$ respectively.  Let $a_0=b_0=0$. 

For each $n$, for $i\ge 1$ if $a_i \notin \bigcup_{j=1}^{i-1} T_n[j]$ then define
$
	T_n[i] = [a_i, a_i + 2^{-n} (b_i - b_{s^{n}})),
$
where $s^{n} = \max_{j<i} \{j : T_n[j] \neq \emptyset\}$ is the most recent non-empty stage, or $s^{n}=0$ if this is the first non-empty stage.  Otherwise define $T_n[i] = \emptyset$.

Let $s^{n}_j$ denote the $j$th non-empty stage, wherever that is well-defined, and let $s^{n}_0 = 0$.  Observe that
\[
	T_n = \bigcup_i T_n[i] = \bigcup_{j\ge1} [a_{s^{n}_j}, a_{s^{n}_j} + 2^{-n} (b_{s^{n}_j} - b_{s^{n}_{j-1}}))
\]
and that all the sets in the above union are disjoint by construction.  As a result
$
	\mu(T_n) = \sum_{j\ge1} 2^{-n} (b_{s^{n}_j} - b_{s^{n}_{j-1}}) \le 2^{-n},
$
so {\rm PA} proves that $(T_n)_n$ is a ML-test.

Because $\alpha$ is provably ML-random, {\rm PA} proves  that there exists an $m$ such  that $\alpha \notin T_m$, so for all $j\ge1$ we know that $s^{m}_j$ is well-defined.  
By construction we have the inequality
$
	a_{s^{m}_{j+1}} \notin [a_{s^{m}_j}, a_{s^{m}_j} + 2^{-m} (b_{s^{m}_j} - b_{s^{m}_{j-1}}))
$
which implies that
$
	b_{s^{m}_j} - b_{s^{m}_{j-1}} \le 2^m(a_{s^{m}_{j+1}} - a_{s^{m}_j}).
$

Defining $a'_j = a_{s^{m}_j}$ and $b'_j = b_{s^{m}_{j-1}}$, we have for all $j\ge1$ that
$
	b'_{j+1} - b'_{j} \le 2^m(a'_{j+1} - a'_{j}),
$
where $(a'_j)_j$ and $(b'_j)_j$ are primitive recursive sequences of rationals which provably converge to $\alpha$ and $\beta$ respectively.  So {\rm PA} proves $\beta \le_S \alpha$.
\end{proof}

\if01
{\bf Comment} With the notation $a'_{i}, b'_{i}$ as below, let

\[
	T'_n[i] = [a'_i, a'_i + 2^{-n} (b'_{i} - b'_{i-1})),
\]

and observe that $\alpha \in T'_n[i] $ implies that $\alpha < a'_i + 2^{-n} (b'_{i} - b'_{i-1})$
which contradicts the construction of $T'_n[i] $ by virtue of which for all $n, i, j >i$ we have
$a'_{j}\not\in T'_n[i]$, hence $a'_{j} > a'_i + 2^{-n} (b'_{i} - b'_{i-1})$ and $\alpha >  a'_{j}$.
Consequently, $\alpha \not\in T'_{1}$, so $\beta \le_S \alpha$ with  constant $c=2$ .
\fi

\begin{cor}\label{cor:dominates}
 If $\alpha$ is c.e.\ and 
provably Chaitin-random and $\beta$ is c.e., then $\beta \le_S \alpha$ is provable in {\rm PA}.
\end{cor}

\begin{proof} See Theorems~\ref{thm:chaitintoML} and \ref{thm:dominates}.
\end{proof}

\if01
\begin{cor}\label{cor:dominates}
If $\alpha \in (0,1)$ is c.e.,
ML-random, and $\beta \in (0,1)$ is c.e., then $\beta \le_S \alpha$.
\end{cor}
\fi

\begin{thm}\label{thm:omegarep}
Suppose $V$ is a provably universal   machine, $\alpha$ is c.e., and $\Omega_V \le_S \alpha$ is provable in {\rm PA}.
Then there exists a provably universal  machine $U$ such that $\Omega_U = \alpha$ is provable in {\rm PA}.
\end{thm}

\begin{proof}


Since $\Omega_V \le_S \alpha$, there exist  primitive recursive increasing sequences $(a_i)_i$ and $(b_i)_i$ of rationals, with limits $\alpha$ and $\Omega_V$ respectively, and a constant $c\ge 0$ such that for all $n$
\begin{equation}
\label{th2:ss}
	b_{n+1} - b_n < 2^c(a_{n+1}-a_n).
\end{equation}
Define $a_0=b_0=0$.  Form the real
\[
	\gamma = \alpha - 2^{-c}\cdot \Omega_V = \sum_n (a_{n+1}-a_n) - 2^{-c}(b_{n+1}-b_n). 
\]
By equation (\ref{th2:ss})  the terms of the sum are positive, so $\gamma \in (0,1)$ is c.e. Applying Lemma~\ref{lem:repce} to $\gamma$ we get a primitive recursive sequence $(m_i)_i$ of natural numbers such that
$
	\sum_i 2^{-m_i} = \gamma.
$

Let $(v_i)_i$ be a 1-1 primitive recursive enumeration of $\dom(V)$, and define the sequence of requests   $	y_{2i}		= v_i ,
	n_{2i}		= |v_i|+c, 
	y_{2i+1}	= v, 		
	n_{2i+1}	= m_i,$
where $v$ is an arbitrarily fixed element in $\dom (V)$.

By Theorem~\ref{thm:kc} we get a provably prefix-free machine $M$ and a primitive recursive enumeration $(x_i)_i$ of $\dom(M)$ such that the following three statements are provable:
1) 	$\mu( \dom(M) ) = \sum_{i} 2^{-n_i}$,
2)	$|x_i| = n_i$ for all $i$,
3)	$M(x_i) = y_i$ for all $i$.

Consider the machine $U = V \circ M$.  The machine $U$ is provably universal. Indeed, $U(x_{2i}) = V(M(x_{2i})) = V(y_{2i})= V(v_{i})$ and $|x_{2i}| = n_{2i}=|v_{i}| + c$,
by 
construction of $M$.
 Finally, it is provable that
$\Omega_U	= \sum_{p\,\in\,\dom(U)} 2^{-p} = \sum_{i} 2^{-n_{2i}} + \sum_{i} 2^{-n_{2i+1}} = 2^{-c}\cdot\Omega_V + \gamma = \alpha. $ 

\end{proof}

Using all results above we obtain:

\begin{thm}
\label{thm:representation} 
A c.e.\ real $\alpha$ is provably Chaitin-random iff it is provable that $\alpha = \Omega_U$ for some provably universal  machine $U$.
\end{thm}
\begin{proof}
Suppose $\alpha$ is provably c.e.\ and Chaitin-random.  By Theorem~\ref{thm:chaitintoML}, $\alpha$ it is provably ML-random.  Take a provably universal  machine $V$ (Theorem~\ref{thm:nonprovunivmachine}).  From 
Theorem~\ref{thm:dominates} we see that $\Omega_V \le_S \alpha$ is provable in {\rm PA}.  By Theorem~\ref{thm:omegarep} we effectively get a $U$ which is provably universal and prefix-free such that $\alpha = \Omega_U$ is provable in {\rm PA}.
The converse is exactly Theorem~\ref{thm:chaitin}.

\end{proof}

\begin{cor}
\label{cor:provchatinrandimpliesrand}
Every provably c.e.\ and Chaitin-random real  is provably random.
\end{cor}
\begin{proof}  If $\alpha$ is provably Chaitin-random and c.e.\  then by Theorem~\ref{thm:representation},  $\alpha = \Omega_U$
for some  provably universal  machine $U$, so $\alpha$ satisfies Solovay's formula (\ref{solovayformula}) with $c=1, \gamma= 1/2 \cdot \Omega_U$.

\end{proof}

 \section{Every Random C.E.\  Real Is Provably C.E.\ Random}

This section proves its title. We start  with the following result  by Solovay \cite{solovayemail}:

\begin{lem}
\label{lem:solovayrepresentation}
Let $V$ be a universal  machine. If $\alpha$ is c.e.\ and ML-random,  
then there exists an integer $c\ge 0$ and a c.e. real $\gamma >0$ such that (\ref{solovayformula}) is satisfied.
\if01
\begin{equation}
\label{solovayrepresentation}
\alpha = 2^{-c} \Omega_{V} + \gamma.
\end{equation}
\fi
\end{lem}

\begin{proof} Using the proof of Theorem~\ref{thm:dominates},  we deduce that $\Omega_{V}\le_{S} \alpha$
(because $\alpha$ is c.e.\ and
ML-random). Consequently, we can consider the primitive recursive increasing sequences $(a_i)_i$ and $(b_i)_i$ of rationals, with $a_0=b_0=0$ and converging to  $\alpha$ and $\Omega_V$ respectively, and a constant $c\ge 0$ such that for all $n$,
$
	b_{n+1} - b_n < 2^c(a_{n+1}-a_n).
$
 The c.e.\ real
$
	\gamma = \alpha - 2^{-c}\cdot \Omega_V = \sum_n (a_{n+1}-a_n) - 2^{-c}(b_{n+1}-b_n) 
$ is positive and $\alpha = 2^{-c} \cdot \Omega_{V} + \gamma.$

\end{proof}

It is not difficult to see that the converse implication in Lemma~\ref{lem:solovayrepresentation} is also true. In fact, a sharper result can be proved:

\begin{thm}
\label{thm:provablyrandom}
Let $V$ be provably universal, $c\ge 0$ be an integer, $\gamma$ a positive c.e.\ real. Then
$\alpha = 2^{-c} \cdot \Omega_{V} + \gamma$ is provably Chaitin-random (ML-random).
\end{thm}
\begin{proof} Let $(v_{n})$ be a primitive recursive enumeration of the domain of $V$ and $(b_{n})_{n}$
be a primitive recursive increasing sequence with limit $\gamma$. The sequence of rationals
$\alpha _{n} = 2^{-c} \sum_{i=1}^{n} 2^{-|v_{i}|} + b_{n}$
is primitive recursive,  increasing and converges to $\alpha$. 

Take $a_{n} = \sum_{i=1}^{n} 2^{-|v_{i}|} $  and observe that for all $n$,
$a_{n+1} - a_{n} \le 2^{c} (\alpha_{n+1} - \alpha_{n}),$
hence {\rm PA} proves that $\Omega_{V} \le_{S} \alpha$.
Using Theorem~\ref{thm:omegarep} we can find  a provably universal  machine $U$ such that $\Omega_U = \alpha$ is provable in {\rm PA}. By Theorem~\ref{thm:chaitin},  $\alpha$ is provably Chaitin-random and by Theorem~\ref{thm:chaitintoML}, $\alpha$ is provably ML-random.
\end{proof}

We can now state our main result:

\begin{thm}
\label{thm: cerandprov}
Every c.e.\ and random real is provably c.e.\ and Chaitin-random (ML-random), hence
provably c.e.\ and random.
\end{thm}
\begin{proof}
Start with a provably universal  machine $V$ (Theorem~\ref{thm:nonprovunivmachine}).  By Lemma~\ref{lem:solovayrepresentation} there exist $c$ and $\gamma$ defining the representation (\ref{solovayformula}) for $\alpha$:
$
	\alpha = 2^{-c} \cdot \Omega_{V} + \gamma.
$
Since $V$ is provably universal, Theorem~\ref{thm:provablyrandom} shows that that $2^{-c} \cdot \Omega_{V} + \gamma$ is provably Chaitin-random  (ML-random).  Therefore $\alpha$ is provably Chaitin-random (ML-random). Finally use Corollary~\ref{cor:provchatinrandimpliesrand} to deduce that $\alpha$ is provably random.
\end{proof}

 Theorem~\ref{thm:representation} can now be stated in the form:

\begin{thm} 
\label{thm:representationfinal} 
A real $\alpha$ is provably  c.e.\  and random iff it is provable that $\alpha = \Omega_U$ for some provably universal  machine $U$.
\end{thm}
\begin{proof} Use  Theorem~\ref{thm: cerandprov} and  Corollary~\ref{cor:provchatinrandimpliesrand}.\end{proof}

\begin{cor} 
\label{cor:uu}
For every universal  machine $U$ there exists  a provably universal  machine $U'$ such that $\Omega_{U} = \Omega_{U'}$. 
\end{cor}
\begin{proof}  Since $\Omega_{U}$  is c.e.\ and random,  by Theorem~\ref{thm: cerandprov}  we deduce that $\Omega_{U}$ is provably Chaitin-random, so by Theorem~\ref{thm:representation} we get a provably universal  machine $U'$ such that $\Omega_{U} = \Omega_{U'}$. 
\end{proof}

\begin{cor} 
\label{cor:secondrep}
Every c.e.\ and random real can be written as the halting probability of a 
 provably universal machine.
\end{cor}
\begin{proof}Use Theorems~\ref{thm: cerandprov} and \ref{thm:representationfinal}.

\end{proof}

\section{A Negative Result}

From the previous two sections we know that every c.e.\ random real can be written
as the halting probability of a provably universal machine, so it is provable random. Does there exist a universal machine whose halting probability is not provable random? By Theorem~\ref{thm:chaitin} such a machine should not be provably universal (and such machines exist by Theorem~\ref{thm:nonprovunivmachine}).

We answer in the affirmative this question.
To this aim we fix an effective enumeration of all c.e.\ reals in (0,1) $(\gamma_{i})_{i}$ 
(for example, by  enumerating all increasing primitive recursive sequences of rationals in (0,1))  and
define the  set $\fR_{\rm ce} = \{\gamma \in (0,1) \mid \gamma \mbox{  is c.e.}\}$.  
 A set $A \subseteq \fR_{\rm ce}$ is called c.e.\ if the set $\{i \in \N\mid \gamma_{i} \in A\}$ is c.e. Note that in $A$ we enumerate all indices for all elements in $A$.

\begin{lem} \label{upp} {\rm \cite{hr}}
If $A \subseteq \fR_{\rm ce}$ is c.e., then for all c.e.\ reals $\alpha \in A$
and $\beta > \alpha$ we have $\beta\in A$.
\end{lem}
\begin{proof} Let $K =\{k_{i}\}$ be a c.e.\ not computable set of natural numbers enumerated by a primitive recursive function $i \mapsto k_{i}$, and for each $n$ let $(a_{i}^{n})^{}_{i}$ be a primitive recursive increasing sequence of rationals
in (0,1) such that $\lim_{i \rightarrow \infty} a_{i}^{n}=\gamma_{n}$. Let $\alpha =\lim_{i \rightarrow \infty} a_{i}^{s}, \beta = \lim_{i \rightarrow \infty} a_{i}^{t}$ and define the function
\[ \Gamma(j,i) = \left\{ \begin{array}{ll}
a_{i}^{s}, & \mbox{\rm if $ j\not=k_{m},  $  for all $m\le i$}, \\
\max\{a_{i}^{s}, a_{i}^{t}\}, & \mbox{\rm otherwise} \,.
  \end{array} \right.\]
Because $\beta > \alpha$ there exists a natural $i_{0}$ such that $a_{i_{0}}^{s} < a_{i_{0}}^{t}$.
If $j\in K$, then there exists an $m$ such that $j=k_{m}$,  hence  $\Gamma(j,i)= a_{i}^{t}$, for 
$i\ge \max\{i_{0}, m\}$, so $\lim_{i \rightarrow \infty} \Gamma (i,j) = \beta$. If $j\not\in K$, then for all $i$, $\Gamma(j,i)= a_{i}^{s}$, so $\lim_{i \rightarrow \infty} \Gamma (i,j) = \alpha$.

Because of the uniform definition of $\Gamma (i,j)$ we can construct a computable function $f$ such that  $\lim_{i \rightarrow \infty} \Gamma (i,j) = \gamma_{f(j)}$. 

Finally, let's assume by absurdity that $\beta\not\in A$. The set $\{j\in \N\mid \gamma_{f(j)} \in A\}$ is c.e.\ because $A$ is c.e., but in view of the definition of $\Gamma$, $\{j\in \N\mid \gamma_{f(j)} \in A\}
= \{j\in \N\mid \gamma_{f(j)}=\alpha\} = \{j\in \N\mid  j\not\in K\}$, a non c.e.\ set.
\end{proof}

Let $(U_{i})_{i}$ be a c.e.\ enumeration of all universal machines. Consider now the sets $\fR_{\rm halt} =\{\Omega_{U_{i}}\}$ and $ \fR_{\rm cerand}^{\rm PA} =\{\gamma \in \fR_{\rm ce} \mid \gamma \mbox{  is provably random}\}.$ By enumerating proofs in PA we deduce that   $ \fR_{\rm cerand}^{\rm PA}$ is c.e., so  
$ \fR_{\rm cerand}^{\rm PA}=\{\gamma_{f(i)}\}$, for some primitive recursive function $f$. 

We have: $\{\gamma_{f(i)}\} \subseteq
\{\Omega_{U_{i}}\} \subset \{\gamma_{i}\}$.
Is $ \fR_{\rm halt}$ c.e.? The answer is negative:

\begin{thm} \label{thm:haltprobnotprov} There exists a universal machine $U_{t}$ such that $\Omega_{U_{t}} \not= \gamma_{f(i)},$ for all $i$.
\end{thm}
\begin{proof} Take a universal machine $U$ such that $\Omega_{U} \ge 1/2$ and construct the c.e.\ real $\beta = \Omega_{U}(n+1)11\cdots $, where $ \Omega_{U}(n+1) = 1^{n}0$. As $\beta > \Omega_{U}$ and $\beta$ is not random, $\beta\not=\Omega_{U_{i}}$,  for all $i$, so by Lemma~\ref{upp}, $\{\Omega_{U_{i}}\}$ is not c.e., hence the  theorem is proved. \end{proof}

{\bf Comment } There is no contradiction between Corollary~\ref{cor:uu} and 
Theorem~\ref{thm:haltprobnotprov}: there exist a universal machine $U_{t}$ and a provably universal machine $U_{j}$ such that  $\Omega_{U_{t}}=\Omega_{U_{j}}$ and  $\Omega_{U_{t}}\not= \gamma_{f(i)},$ for all $i$:
 PA cannot prove the randomness of $\Omega_{U_{t}}$ based on $U_{t}$, but can prove the randomness of $\Omega_{U_{t}}=\Omega_{U_{j}}$ based on $U_{j}$.

\section{Formal Proof of the Kraft-Chaitin Theorem}


In the above we gave proofs that various statements, once suitably formalised in the language of first order logic, were derivable from the axioms of PA.  In principle, but for lack of space and patience, we could have presented complete PA derivations of each statement proved.  Instead, as is common practice for all but the simplest of results, we sketched constructions which leave the actual derivations implicit.

Recent advances in theorem proving computer programs, such as the proof assistant Isabelle \cite{isabelletut}, have allowed complete formal derivations of nontrivial mathematical results.  In such systems, humans write a sequence of proof commands, and the computer system searches for a complete derivation, if one exists.  Essentially, the human user gives a sequence of intermediate lemma with proof directions, and the computer interpolates the full derivation.  (For a recent perspective on the importance of formalising mathematics see \cite{TCH:FP}.)

Using Isabelle, we formalised and proved the Kraft-Chaitin Theorem (Theorem \ref{thm:kc}), a key result in our above proof.  To keep our presentation self-contained, we begin by showing how to formalise and prove a simple result about strings; for a full introduction to the Isabelle system see \cite{isabelletut}.  We follow with a formalisation  of the Kraft-Chaitin Theorem, then sketch its formal proof.  The full proof script is available online \cite{kcproof}.


\subsection{Formalising Results in Isabelle}

To illustrate Isabelle and its use, we will formalise and prove the following simple property of strings:

\begin{lem}
Given strings $x,y,z \in \Sigma^*$, if $x$ extends $y$ then $xz$ extends $y$.
\end{lem}

Strings are naturally represented by the Isabelle list data-type.  Here {\tt []} represents the empty list, and {\tt y\#ys} represents the list formed by concatenating the element {\tt y} with the list {\tt ys}.  For example, the string $001$ is represented by {\tt 0 \# 0 \# 1 \# []} (or {\tt [0,0,1]} for short).
The following code inductively defines whether the list {\tt A} extends {\tt B}, denoted {\tt extends A B}:

\begin{verbatim}
fun extends :: "'A list => 'A list => bool"
where
  "extends [] [] = True"
| "extends [] (y#ys) = False"
| "extends x [] = True"
| "extends (x#xs) (y#ys) = ((x=y) & (extends xs ys))"
\end{verbatim}

When faced with the above definition, Isabelle automatically proves termination  (in this case, by observing that the first argument always decreases in length with each recursive call).

Let us first prove that any list extends the empty list.  We enter into Isabelle:

\begin{verbatim}
lemma extends1: "extends A []"
\end{verbatim}

It responds with the propositions we need to prove:

\begin{verbatim}
goal (1 subgoal):
 1. extends A []
\end{verbatim}

It is natural to prove this by induction on {\tt A}, by entering the command \verb|apply(induct A)|.  This results in two proof obligations, one for the base case and the other for the inductive step:

\begin{verbatim}
goal (2 subgoals):
 1. extends [] []
 2. !!a A. extends A [] ==> extends (a # A) []
\end{verbatim}

The first proposition is one of the cases in our definition of {\tt extend}.  In the second {\tt !!} denotes universal quantification and this similarly follows from one of our definition cases.  We tell Isabelle to simplify these expressions with the command \verb|apply(simp_all)|. Isabelle manages to simplify all these expressions down to {\tt True}, using rewrite rules for simplifying conjunctions, variable identity, and expanding the definition of {\tt extends}.  As a result we get:

\begin{verbatim}
goal:
No subgoals!
\end{verbatim}

Having completed the proof, we compactly store it in the following format:

\begin{verbatim}
lemma extends1: "extends A []"
  apply(induct A) apply(simp_all)
done
\end{verbatim}

We can now attempt our original goal:

\begin{verbatim}
lemma extends2: "extends (A@B) A"

goal (1 subgoal):
 1. extends (A @ B) A
\end{verbatim}

The concatenation of lists {\tt A} and {\tt B} is denoted {\tt A @ B}.  We again induct with the command \verb|apply(induct A)|, then simplify with the command \verb|apply(simp_all)|

\begin{verbatim}
goal (2 subgoals):
 1. extends ([] @ B) []
 2. !!a A. extends (A @ B) A ==> extends ((a # A) @ B) (a # A)

goal (1 subgoal):
 1. extends B []
\end{verbatim}

Since we proved this before, we use the command \verb|apply(simp only: extends1)| to reuse our previous result, completing the proof.  In sum:

\begin{verbatim}
lemma extends2: "extends (A@B) A"
  apply(induct A) apply(simp_all) apply(simp only: extends1)
done
\end{verbatim}

\subsection{Formalising the Kraft-Chaitin Theorem}
\if01
\begin{thm}
\label{thm:kc2}
Suppose $(n_i, y_i)_i \in \N \times \Sigma^*$ is a primitive recursive enumeration of requests such that
\[
	\sum_{i} 2^{-n_i} \le 1.
\]
Then there exists a provably prefix-free machine $M$ and a primitive recursive enumeration $(x_i)_i$ of $\dom(M)$ such that the following is provable in {\rm PA}:
\begin{enumerate}
\item	$\mu( \dom(M) ) = \sum_{i} 2^{-n_i}$,
\item	$|x_i| = n_i$ for all $i\in\N$,
\item	$M(x_i) = y_i$ for all $i\in\N$.
\end{enumerate}
\end{thm}

\begin{algorithm}[htb]
\caption{}\label{alg:kc2}
\begin{algorithmic}[1]

\State $S_0 = \{\epsilon\}$, $T_0 = \emptyset$, $r_0 = 0$, $i \setto 0$.
\Loop
	\State	Let $s_i$ be the longest element of $S$ of length at most $n_i$.  If no such string exists, terminate.
	\If{$|s_i|=n_i$}
		\State	$S_{i+1} = (S_i \setminus \{s_i\})$.
	\Else
		\State	$S_{i+1} = (S_i \setminus \{s_i\}) \cup \{ s_i1, s_i01, s_i0^21, \ldots, s_i0^{n_i-|s_i|-1}1 \}$.
	\EndIf
	\State	Define $M(s_i0^{n_i-|s_i|}) = y_i$.
	\State	$T_{i+1} = T_i \cup \{s_i0^{n_i-|s_i|}\}$.
	\State	$r_{i+1} = r_i + 2^{-n_i}$.
	\State	$i\setto i+1$.
\EndLoop

\end{algorithmic}
\end{algorithm}
\fi

The proof of the Kraft-Chaitin Theorem  is algorithmic:  it describes a particular algorithm (Algorithm~\ref{alg:kc} of Theorem~\ref{thm:kc}) for selecting strings of the required lengths, and proves that the algorithm is correct.  In what follows we will  implement this algorithm in Isabelle and will prove its correctness.



The following Isabelle code implements Algorithm \ref{alg:kc}.  We  give the definition of each function, then explain what it does.

\begin{verbatim}
fun extend :: "nat list => nat => nat list list"
where
  "extend l 0 = [l]"
| "extend l (Suc n) = (hd (extend l n) @ [0]) # (hd (extend l n) @ [1]) 
                                              # tl (extend l n)" 
\end{verbatim}

For {\tt l} a binary list representing a binary string, and {\tt n} a natural number, {\tt extend l n} computes the list
\[
	[ l0^{n-|l|}, l0^{n-|l|-1}1, \dots, l01, l1 ].
\]
For example, in Isabelle the expression {\tt extend [0,0,1] 5}  evaluates to 
\begin{center}
	{\tt [[0,0,1,0,0], [0,0,1,0,1], [0,0,1,1]] }
\end{center}
This corresponds to the set $\{ 00100, 00101, 0011\}$ of binary strings.

The set of unallocated prefixes $S_i$ and  the set of allocated strings $T_i$ are represented by lists of strings.  The free prefixes are ordered by decreasing length, the allocated strings by the order of allocation.

Consider one iteration of the main loop.  Let {\tt A} be the list of previously allocated strings, {\tt F} the list of free prefixes, and {\tt n} the length of the string we want to allocate at  this step.  (These are denoted  $T_i$, $S_i$, and $n_i$ in the original algorithm.)
{\tt kcstep A F n} returns the updated pair of allocated strings and free prefixes ($S_{i+1}$ and $T_{i+1}$).

\begin{verbatim}
consts kcstep :: "nat list list => nat list list => nat
                                => (nat list list * nat list list)"
primrec
  "kcstep A [] n       = (A, [])"   (* fail case *)
  "kcstep A (f # F) n  = (if length f <= n 
                then ((hd (extend f (n - length f))) # A, 
                      (tl (extend f (n - length f))) @ F)
                else (fst (kcstep A F n), f # snd (kcstep A F n)))"
\end{verbatim}

{\tt kcstep} searches through the list {\tt F} of free prefixes for the longest string of length at most {\tt n}.  One it finds it, it calls {\tt extend}, which returns a list of extended prefixes.  It takes the first string in the list, guaranteed to have  exactly length {\tt n}, and adds it to the allocated strings list.  The rest of the strings are placed on the free prefixes list.

For example, {\tt kcstep [] [[]] 2} evaluates to 
\begin{center}
	{\tt ([[0,0]], [[0,1], [1]]) }
\end{center}
which corresponds to the list $00$ of allocated strings and the set $\{1, 01\}$ of free prefixes.

\begin{verbatim}
consts kcloop :: "nat list => (nat list list * nat list list) 
                           => (nat list list * nat list list)"
primrec
  "kcloop [] X = X"
  "kcloop (l#ls) X = (kcstep (fst (kcloop ls X)) (snd (kcloop ls X)) l)"
\end{verbatim}

For a list of lengths {\tt l} and a pair {\tt (A,F)} of allocated strings and free prefixes, {\tt kcloop l (A,F)} runs {\tt kcstep} to allocate strings  for every length in {\tt l}.  For example, {\tt kcloop [3,4,2] ([], [[]])}  allocates a string of length 2, then one of length 4, then one of length 3, starting from the initial state where no strings are yet allocated ({\tt []}) and the empty string is our free prefix ({\tt [[]]}).

For example, {\tt kcloop [3, 2] ([], [[]])} evaluates to 
\begin{center}
	{\tt ([[0,1,0], [0,0]], [[0,1,1], [1]]) }
\end{center}
which corresponds to the list $00, 010$ of allocated strings (note that we reverse the list), and the set $\{1, 011\}$ of free prefixes.

\begin{verbatim}
fun kc :: "nat list => nat list list"
where
	"kc ls = (fst (kcloop ls ([],[[]])))"
\end{verbatim}

For a list of lengths {\tt l}, {\tt kc l} returns the list of strings allocated by running {\tt kcloop} on the list starting from the initial state where no strings have been allocated.  For example, {\tt kc [4,3,2]} evaluates to
\begin{center}
	{\tt [[0,1,1,0], [0,1,0], [0,0]] }
\end{center}
which corresponds to the sequence $00, 010, 0110$ of allocated strings.

This implements Kraft-Chaitin's algorithm, for we will prove that:
\begin{enumerate}
\item	If our list of lengths obeys Kraft's inequality, $\sum_{i} 2^{-n_i} \le 1$, then {\tt kc ls} is a list of strings, and the $i$th element of {\tt kc ls} has length equal to the $i$th element of {\tt ls}.
\item	{\tt kc ls} is always a prefix-free list (no two distinct elements of the list are prefixes of each other).
\item	If we add new lengths to the start of {\tt ls}, then this adds new strings to the end of {\tt kc ls} without changing the old ones.  That is, once a string of a given length is allocated it is not changed.
\end{enumerate}

To prove the above we need to define what a prefix-free list is, a function to evaluate Kraft's inequality, a function which checks whether the lengths of one list match the lengths in another, and a tool to check whether one list extends another.


\begin{verbatim}
fun prefixes :: "nat list => nat list => bool"
where
  "prefixes [] x = True"
| "prefixes x [] = True"
| "prefixes (x#xs) (y#ys) = ((x=y) & (prefixes xs ys))"

consts incomparable :: "nat list => nat list list => bool"
primrec
  "incomparable x [] = True"
  "incomparable x (y # ys) = (~(prefixes x y) & (incomparable x ys))"

consts prefixfree :: "nat list list => bool"
primrec
  "prefixfree [] = True"
  "prefixfree (x # xs) = ((incomparable x xs) & (prefixfree xs))"
\end{verbatim}

If {\tt x} is a prefix of {\tt y}, or vice versa, then {\tt prefixes x y}. For example {\tt prefixes [0,0,1] [0,0]} is true.
{\tt incomparable x A} holds if {\tt x} is not a prefix of any string in {\tt A}, for instance {\tt incomparable [0,0] [[1,0], [1,1,1]]} holds.
{\tt prefixfree L} holds if the list {\tt L} is prefix-free, for instance {\tt prefixfree [[0,0], [1,0], [1,1,1]]} holds.

\begin{verbatim}
consts expn2 :: "nat => rat"
primrec
  "expn2 0 = 1"
  "expn2 (Suc n) = (1/2) * expn2 n"

consts meas_nat :: "nat list => rat"
primrec
  "meas_nat [] = 0"
  "meas_nat (f#F) = (expn2 f + meas_nat F)"
\end{verbatim}

We define {\tt expn2 n} equal to $2^{-n}$.  {\tt meas\_nat F} computes the ``measure'' of a sequence of natural numbers {\tt F},
for example {\tt meas\_nat [4,3,2]} equals {\tt 7/16}.

\begin{verbatim}
fun lengthsmatch :: "nat list list => nat list => bool"
where
  "lengthsmatch [] [] = True"
|  "lengthsmatch [] (l#ls) = False"
|  "lengthsmatch (x#xs) [] = False"
|  "lengthsmatch (x#xs) (l#ls) = ((length x = l) & (lengthsmatch xs ls))"
\end{verbatim}

The expression {\tt lengthsmatch X Y} holds if the lengths of each string in {\tt X} matches the corresponding number in {\tt Y}.
For example, we have {\tt lengthsmatch [[0,0], [1,0], [1,1,1]] [2,2,3]} is {\tt True}.

\begin{verbatim}
fun extends :: "'A list => 'A list => bool"
where
  "extends [] [] = True"
| "extends [] (y#ys) = False"
| "extends x [] = True"
| "extends (x#xs) (y#ys) = ((x=y) & (extends xs ys))"
\end{verbatim}

Finally, {\tt extends A B} holds if the list {\tt A} extends the list {\tt B}, so {\tt extends [0,1] [0]} holds.

With the above definitions we can state the three results which establish correctness:

\begin{verbatim}
theorem kc_correct1: "meas_nat ls <= 1 ==> lengthsmatch (kc ls) ls"

theorem kc_correct2: "prefixfree (kc ls)"

theorem kc_extend:   "extends (rev (kc (L2 @ L1))) (rev (kc L1))"
\end{verbatim}

The first says that if {\tt ls} is a list of natural numbers $n_i$ which satisfies Kraft's inequality $\sum_i 2^{-n_i} \le 1$,
then the strings {\tt kc  ls} allocated by running Algorithm~\ref{alg:kc} on this list have exactly the lengths {\tt ls} we asked for.

The second says the strings allocated are prefix-free.

The last says that when Algorithm~\ref{alg:kc} allocates additional strings it does not change strings it has previously allocated.  To see this, note that when we run {\tt kc L} the algorithm allocates strings starting from the end of the list {\tt L}.  This means, the first element of {\tt kc L} is the last string allocated.  {\tt kc (L2 @ L1)} is the list of strings allocated if we allocate strings with lengths in {\tt L1} then strings with lengths in {\tt L2}.

Together, establishing these would show that the {\tt kc} algorithm constructively establishes the Kraft-Chaitin Theorem.

\subsection{Proof Outline}

All the above merely formalised the algorithm and stated the theorem we wish Isabelle to prove.  This gets the order mixed slightly, since formalising this theorem unearthed a mistake in the algorithm, so the process was mutual.  In some sense this formalisation of the theorem is the major creative work, the rest is just technical detail.  As one might guess, however, most of the work is in these details.  To prove the above theorems we must guide Isabelle to them by establishing numerous intermediate lemmas, and telling Isabelle which proof techniques to use to establish each.  Often we just advise Isabelle to induct on a variable then simplify, but sometimes we must give more detailed guidance.


The Isabelle proof follows the proof given for Theorem~\ref{thm:kc}:  we establish that the inner loop preserves some invariants, and use these invariants to establish correctness.

Recall the algorithm has two variables:  the list of allocated strings and the list of free strings.  Each pass through the loop will (potentially) add one new allocated string, and modify the free strings.  We then show that these two lists combined  remain prefix-free, their joint measure never decreases, and that there are never two free strings of the same length.

For reasons of space we give only the definitions required to state the above intermediate results and show how they are formalised in Isabelle.  The proof in its entirety is available online \cite{kcproof}.

\begin{verbatim}
fun strictlysorted :: "nat list list => bool"
where
  "strictlysorted [] = True"
| "strictlysorted [x] = True"
| "strictlysorted (x1 # x2 # xs) = ((length x1 > length x2) 
                                   & (strictlysorted (x2 # xs)))" 
\end{verbatim}

{\tt strictlysorted L} holds if the strings in {\tt L} are ordered by (strictly) decreasing length.  In particular, this means there can be no two strings of the same length in {\tt L}.

\begin{verbatim}
fun inv1 :: "nat list list * nat list list => bool"
where
  "inv1 X = strictlysorted (snd X)"

fun inv2 :: "nat list list * nat list list => bool"
where
  "inv2 X = prefixfree ((fst X) @ (snd X))"

\end{verbatim}

The first invariant is that the list of free strings is strictly sorted.  This is needed to show both that there is at most one string of any given length and to show that the algorithm will always select the longest string it is able to.

\begin{verbatim}
fun inv :: "nat list list * nat list list => bool"
where
  "inv X = ((inv1 X) & (inv2 X))"

theorem kcstep_inv: "inv (A,F) ==> inv (kcstep A F n)"
\end{verbatim}

This says simply that if the invariants held of the variables before running through the loop once, then they hold afterwards.

\begin{verbatim}
consts meas :: "nat list list => rat"		(* The measure of a prefix free set *)
primrec
  "meas [] = 0"
  "meas (x # xs) = expn2 (length x) + meas xs"
\end{verbatim}

This defines the measure of a list of strings: the usual $\sum_{x\in X} 2^{-|x|}$.

\begin{verbatim}
lemma kcstep_meas: "meas ((fst (kcstep A F n)) @ (snd (kcstep A F n))) 
                    = meas (A@F)"
\end{verbatim}

This says that measure is preserved at each step of the loop.  This measure will be 1 for all the intermediate states of the {\tt kc} algorithm, but we need this more general result for the inductive proofs to work.

A number of further intermediate results are required both to establish the above invariants and to apply them to the main theorems.  Below are three of the most important, which one may recall from the proof of Theorem \ref{thm:kc} (in total, there are 102 theorems and lemmas proved).

\begin{verbatim}
theorem meas_alloc: "[| expn2 n <= meas F; strictlysorted F |] 
                     ==> length (last F) <= n"

lemma kcstep_correct1: "[|inv (A,F); expn2 n <= meas F|]
            ==> (tl (fst (kcstep A F n)) = A) 
              & (length (hd (fst (kcstep A F n))) = n)"          

lemma kcstep_correct2: "[|inv (A,F); expn2 n <= meas F|] 
            ==> meas (fst (kcstep A F n)) = meas A + expn2 n"
\end{verbatim}

Theorem {\tt meas\_alloc} formalises the result that if $2^{-n}$ is smaller than the measure of a set $F$, and that set has no two strings of the same length, then there is a string of length at most $n$ in $F$.
Lemma {\tt kcstep\_correct1} says that if the invariants are satisfied by the current variables {\tt A} and {\tt F}, and the measure of {\tt F} is at least $2^{-n}$, then kcstep succeeds.  This means that we allocate one new string of length exactly $n$, leaving the old strings untouched.
Lemma {\tt kcstep\_correct2} expresses an implied result:  if the algorithm succeeds, then the measure of the list of allocated strings increases by exactly $2^{-n}$ (Isabelle will often not notice conclusions that seem obvious to the prover; they must be spelt out).

\section{Final Remarks}

If {\rm PA}  receives an algorithm for a machine  $V$, a proof that $V$ is universal and prefix-free, an
 integer $c\ge 0$, and a computable increasing sequence of rationals converging to a real $\gamma>0$, then  {\rm PA} can  prove that $\alpha = 2^{-c}\cdot\Omega_V + \gamma$ is c.e.\ and random.  Similarly,
if {\rm PA} receives an algorithm for a machine  $U$, a proof that $U$ is universal and prefix-free, then it can prove that $\Omega_U$ is c.e.\ and random.
This implies  that every c.e.\ random real is provably c.e.\ and random---as stated in
Solovay's email \cite{solovayemail}.

We have offered two representations for c.e.\ and random reals from which {\rm PA} can prove
that the real is  c.e.\ and random. In the first we fix a provably universal prefix-free machine $V$ and we vary the integer $c\ge 0$ and the c.e. real $\gamma>0$ to get via the formula $ 2^{-c}\cdot \Omega_V + \gamma$  all c.e.\ and random reals. In the second we vary all provably universal prefix-free machines $U$ to get via $\Omega_{U}$ all c.e.\ and random reals.

A key result was to
show that the  theorem  that ``a  real is c.e.\ and random iff it is the halting probability of a universal   machine'' \cite{CHKW98stacs,slaman,Ca} can be proved  in {\rm PA}.   Our proof, which is simpler than the standard one, can be  used also for the original theorem.

We proved two negative results showing the importance of the representation for provability of randomness: a) there exists a universal machine whose universality cannot be proved in PA,  b) there exists a universal machine $U$ such that, based on $U$,  PA cannot prove the randomness of $\Omega_{U}$.

Chaitin \cite{chaitin98} explicitly computed a  constant $c$ such that if $N$ is larger than the size in bits of the program 
for enumerating the theorems of {\rm PA} plus $c$,
then  {\rm PA} cannot prove that a specific  string  $x$ has complexity greater than $N$, $H_{U}(x)>N$. Consequently, {\rm PA} cannot prove randomness of almost all random (finite) strings.  Our positive result shows an interesting difference between  the finite and the infinite cases of (algorithmic) randomness.

Does our positive result contradict Chaitin and  Solovay's negative results discussed in the Introduction? The answer is negative because
the digits of the binary expansion of a random c.e.\ real are not computable.

Our positive result 
would not be satisfactory without demonstrating our proofs with an automatic theorem prover.
We have chosen Isabelle \cite{isabelletut} to obtain an automatic proof of our version of the Kraft-Chaitin Theorem, one of the key results of this paper.  The paper contains a description of the formalisation (for Isabelle) of the Kraft-Chaitin Theorem and the description of the main steps of the automatic  proof; the full proof is available online 
\cite{kcproof}. 

Finally we speculate about the role of the automatic prover.  How can an automatic theorem prover help understanding/proving a mathematical statement?\footnote{The reader may note that we don't question the fact that an automatic theorem prover {\em  helps} understanding mathematics, \cite{CCM}.}  There are at least three possibilities. a)
Use the prover to verify
the theorem by discovering a proof, call it ``Solovay mode'' (because this corresponds to
the result reported in this paper: Bob Solovay communicated to one of us the statement to be proved and we  found a proof). It is worth observing that the Kraft-Chaitin Theorem has two ``roles'':  one, as an algorithm,  to be executed, the
other, as a mathematical statement,  to be  proved.  Previous
formalisation efforts focused only on the first part\footnote{The use of  Lisp and Mathematica in Algorithmic Information Theory were pioneered by Chaitin---see  \cite{chaitin98}.}; our present work was directed towards the second.  One could imagine that mathematical journals might use such systems in the process of refereeing \cite{CCM,PF}.
b) The second possibility is to use the
prover to verify a human-made proof---a full Isabelle proof for all results in this paper is under construction. c) The third possibility is to use the prover as  some kind of  ``assistant'' in an interactive process of discovery/proving. During the work to automate the proof  of the Kraft-Chaitin Theorem a mistake in our human-made argument was unearthed and corrected.
We also used the experience with Isabelle to test the adequacy  of the representation of a c.e.\ random real in meeting the goal: to obtain the {\rm PA} proof of randomness.



\section*{Acknowledgment}

We thank Bob Solovay for suggesting the result of this paper and useful comments,  Jeremy Dawson for  helpful advice on the Isabelle proof, and Greg Chaitin, Liam Fearnley, Bruno Grenet,  Mathieu Hoyrup, Andr\'e Nies, 
Cristobal Rojas, Frank Stephan, Garry Tee and the anonymous referee for useful comments which improved our paper.


\begin{thebibliography}{99}
\bibitem{BL} W. S. Brainerd, L. H. Landweber. {\em Theory of Computation}, Wiley, New York, 1974.
\bibitem{Incompl} C.~S.~Calude. Chaitin $\Omega$ numbers, Solovay machines
and incompleteness, 
 {\em
Theoret. Comput. Sci.}  284 (2002), 269--277.

\bibitem{cerand} C. S. Calude. A characterization of c.e.~random reals, {\em Theoret. Comput. Sci.}
271 (2002), 3--14.
\bibitem{Ca} C. S. Calude. {\em Information and Randomness. An Algorithmic
    Perspective}, 2nd Edition, Revised and Extended, Springer Verlag, Berlin,
  2002.
  
  \bibitem{CCM}  C. S. Calude, E. Calude, S. Marcus. Proving and Programming, in C. S. Calude (ed.).  {\em Randomness \& Complexity, from Leibniz to Chaitin},
World Scientific, Singapore, 2007, 310--321.
  
  
  \bibitem{CHKW98stacs}
C.~S.~Calude, P.~Hertling, B.~Khoussainov, and Y.~Wang.
Recursively enumerable reals and {C}haitin $\Omega$ numbers,
in: M.~Morvan, C.~Meinel, D.~Krob \ (eds.),
{\em Proceedings of the 15th Symposium on
Theoretical Aspects of Computer Science (Paris)},
Springer--Verlag, Berlin, 1998, 596--606. Full paper  in {\em Theoret.
Comput. Sci.}  255 (2001), 125--149.


  \bibitem{chaitin75}
G. J. Chaitin. A theory of program size formally identical to
information theory, {\em J. Assoc. Comput. Mach.} 22 (1975), 329--340.

  \bibitem{chaitin87} G. J. Chaitin. {\em Algorithmic Information Theory},
  Cambridge University Press, Cambridge, 1987 (3rd printing 1990).
  
  
   \bibitem{chaitin98} G. J. Chaitin. {\em The Limits of Mathematics}, Springer,
   Singapore, 1998.
   
   \bibitem{TCH:FP}
T.~C. Hales. Formal proof, {\em Notices of the AMS} 11 (2008), 1370--1380.
   
  \bibitem{DH} R. Downey, D. Hirschfeldt. {\em Algorithmic Randomness
    and Complexity}, Springer, Heidelberg,  to appear.
    
    
    \bibitem{PF} P. C. Fischer.  Theory of provable recursive functions, {\em Trans. Amer.
    Math. Soc.} 117 (1965), 494--520.
    
    
\bibitem{kcproof}
N. J. Hay.  Formal proof of the Kraft-Chaitin theorem in Isabelle.  Available online at \url{http://www.cs.auckland.ac.nz/~nickjhay/KraftChaitin.thy}.

\bibitem{hr} M. Hoyrup, C. Rojas. Personal communication to C. Calude, 11 September 2008.
    
    \bibitem{RK} R. Kaye. {\em Models of Peano Arithmetic}, Oxford Press, Oxford, 1991.
    
    
\bibitem{slaman} 
 A. Ku\v{c}era, T.~A. Slaman. Randomness and recursive enumerability, {\em SIAM J. Comput.},
31, 1 (2001), 199-211.



\bibitem{isabelletut}  T. Nipkow,  L. C. Paulson,  M. Wenzel.  {\em Isabelle/HOL --- A Proof Assistant for Higher-Order
Logic}, Springer,  LLNCS 2283, 2002.


  \bibitem{solovayemail} R. M. Solovay. Personal communication to C. Calude, 23 March 2007.



\bibitem{solovay2k}
R. M. Solovay.  A version of $\Omega$ for which  {\rm ZFC} can not predict a
single bit, in C.S. Calude, G. P\u{a}un (eds.). {\em Finite Versus
Infinite. Contributions to an Eternal Dilemma}, 
Springer-Verlag, London, 2000, 323--334.

\bibitem{solovaymanu}
R. M. Solovay.
{\em Draft of a paper (or series of
 papers)  on Chaitin's work \ldots  done for the most
part during the period of Sept.--Dec. 1974}, unpublished manuscript, IBM
Thomas J. Watson Research Center, Yorktown Heights, New York, May 1975,
215 pp.



\end{thebibliography}
\end{document}